\documentclass[letterpaper, 10 pt, conference]{ieeeconf}  

\IEEEoverridecommandlockouts                              
\usepackage{tikz}
\usepackage{amssymb}
\usepackage[utf8]{inputenc} 
\usepackage[T1]{fontenc}    
\usepackage{hyperref}       
\usepackage{url}            
\usepackage{booktabs}       
\usepackage{amsfonts}       
\usepackage{nicefrac}       
\usepackage{microtype}      
\usepackage{caption}
\usepackage{subcaption}
\usepackage{verbatim}
\usepackage{xcolor}
\usepackage{bbm}
\usepackage{bm}
\usepackage{graphicx}
\usepackage{pgf}
\usepackage{amsmath}
\usepackage{cleveref}
\usepackage{mathtools}

\newcommand{\expect}[1]{\mathbb E\left[#1\right]}
\newcommand{\expectsub}[2]{\mathbb E_{#1}\left[#2\right]}
\newtheorem{assumption}{Assumption}
\newtheorem{proposition}{Proposition}

\newtheorem{theorem}{Theorem}

\newtheorem{lemma}{Lemma}
\newtheorem{remark}{Remark}
\newtheorem{algorithm}{Algorithm}

\title{\LARGE \bf
Labor-right Protecting Dispatch of Meal Delivery Platforms
}

\author{Wentao Weng and Yang Yu 
\thanks{W. Weng is with the Department of Electrical Engineering and Computer Science,
        Massachusetts Institute of Technology
       {\tt\small wweng@mit.edu}. 
Y. Yu is with the Institute for Interdisciplinary Information Sciences,
        Tsinghua University
        {\tt\small yangyu1@mail.tsinghua.edu}.%
        This work was done when W.Weng was a senior undergraduate at Tsinghua University.
}
}

\begin{document}
\maketitle
\thispagestyle{empty}
\pagestyle{empty}

\begin{abstract}
The boom in the meal delivery industry brings growing concern about the labor rights of riders. Current dispatch policies of meal-delivery platforms focus mainly on satisfying consumers or minimizing the number of riders for cost savings. There are few discussions on improving the working conditions of riders by algorithm design. The lack of concerns on labor rights in mechanism and dispatch design has resulted in a very large time waste for riders and their risky driving. In this research, we propose a queuing-model-based framework to discuss optimal dispatch policy with the goal of labor rights protection. We apply our framework to develop an algorithm minimizing the waiting time of food delivery riders with guaranteed user experience. Our framework also allows us to manifest the value of restaurants’ data about their offline-order numbers on improving the benefits of riders.
\end{abstract}

\section{INTRODUCTION}
\subsection{Background}
Labor rights in meal delivery have become one of the top public concerns in platform economy. During the last decade, the meal-delivery industry has been flourishing due to the penetration of smart mobile phones. However, a growing number of traffic violations and accidents due to risky driving of riders have attracted people to discuss the relationship between platform algorithms and labor rights. Current statistics reveal that meal delivery riders contribute to a significant share of traffic violations. In 2017, approximately one food delivery rider was reported to be injured or killed every 2.5 days in Shanghai, China\cite{ChinaNews}. The emerging food-delivery platform also contributes to a nonnegligible fraction of traffic accidents in Korea \cite{byun2017characteristics}. In the U.S., food delivery riders are viewed as one of the most dangerous jobs\cite{DangerousJob}.

It is necessary to develop and analyze dispatch algorithms of meal-delivery platforms from the labor-right protecting perspective. However, the current algorithmic models for meal delivery are either user-centric or platform profit-driven and ignore the perspective of protecting riders from traffic risks. The absence of the perspective of labor rights of riders causes two problems. First, the current dispatch models simplify the procedure of the working flow of riders by merging riders’ picking-up order period and delivery period into one period. This oversimplification causes the current models to ignore the potential of improving rider traffic safety by managing the rider waiting time during the period of picking up orders. Consequently,  user-centric or profit-driven dispatch results in overwaiting riders. Second, the user-centric dispatch focuses mainly on timely delivery for attracting consumers while profit-driven dispatch pursues minimizing the population of riders for cost saving\cite{DriverRisk, christie2019health}. Consequently, the current dispatch algorithms limit the number of riders and continuously shorten the time window of deliveries to attract more consumers by quickly delivering foods. Both mechanisms, user-centric and profit-driven dispatch, contribute to overload of riders. The above two problems of the current food-delivery models together drive riders to overspeed.

In this research, we propose a queuing-model-based framework of designing dispatch policies from the perspective of protecting the traffic safty of riders in meal delivery platforms. In contrast to user-centric dispatch or profit-driven dispatch, in this research, we develop a framework that carefully models the working flow of riders and seeks to mitigate the stress of the delivery-time limit that forces riders to drive at excessive speeds. Our rider-centric framework enables us to analyze the pathway of protecting rider safety. According to our framework, we develop a labor-rights-protection dispatch algorithm to control the waiting time of riders in restaurants. Our algorithm can achieve the optimal dispatch strategy that minimizes the waiting time of riders with guaranteed consumer-acceptable order waiting time. We also apply our framework to explore how restaurant information can further improve the labor rights of riders with a refined platform dispatch policy.

Our rider-centric framework further allows us to comprehensively capture the impacts from all stakeholders, including restaurants, consumers, platforms, and riders themselves. We found that information about the number of restaurant offline orders is critical for managing rider waiting time for orders. Therefore, our framework decomposes the meal-delivery task into two processes: an information disclosure process and a dispatch process. The information disclosure process captures the impact of restaurants sharing their information, while the dispatch process captures the impact of the patience of consumers and decisions of platforms. We found that restaurants can help with labor rights protection by sharing their private information about the number of their offline orders. This conclusion reveals the tight connection between private information and labor rights in platform economy.
\subsection{Related work}

The growth of meal delivery platforms opens a wide range of research areas. Due to the unique process structure of meal deliveries, queuing models are frequently used to understand profit brought by meal delivery platforms, or broadly on-demand platforms \cite{bai2019coordinating,taylor2018demand,feldman2019can,chen2019food,chen2020courier}. The current user-centric or profit-driven dispatch adopts mainly a single-sided queueing model focusing on timely service of the randomly arriving orders by the least size of rider population. In particular, \cite{feldman2019can} studies the benefit of delivery platforms to restaurants by analyzing an associated queuing model and proposing new contracts. Then, in \cite{chen2019food}, the authors propose a queuing model with a pricing scheme to study interactions between platforms and customers. Their model also discusses the impact of a finite rider pool on the profit of a platform. In addition, the model of \cite{chen2020courier} analyzes the benefit of order batching via queues with batch arrivals.

From a technical perspective, our queueing model belongs to the class of two-sided queues, whose study can be dated back to the 1990s \cite{gelenbe1991queues}. In the original model, there are arrivals of positive jobs to one single queue. There are also negative jobs that can cancel one positive job waiting in the queue. However, when there is no positive job, the negative job will directly leave. The sojourn time of positive jobs has been derived \cite{harrison1993sojourn}, and the model has also been extended to more general service time \cite{harrison1996m} and arrival processes \cite{li2004map}. The stationary distribution is known to have a product form in queueing networks with negative arrivals \cite{henderson1993queueing, gelenbe1991product}. Our model, however, is different from previous studies in the sense that negative arrivals, i.e., riders, can wait in another queue for future orders. The arrival process of riders can depend on the queue length of orders instead of a time-homogeneous process.

In addition to queuing models, there are optimization and reinforcement learning studies aiming at developing online optimization strategies for efficient real-time meal deliveries. A stream of work models rider dispatch problems by optimizations with the objective of maximizing the number of timely orders or minimizing the averaged delays \cite{yildiz2019provably, reyes2018meal,joshi2020batching}. In addition, studies of meal deliveries when a customer's order can include food from multiple restaurants is studied in \cite{steever2019dynamic}. Randomness in meal delivery problems is considered in the optimization framework in \cite{ulmer2020restaurant,liu2020time}. The model in \cite{liu2020integrating} uses deep inverse reinforcement learning to learn a better dispatch policy. In addition to providing efficient dispatch, customer satisfaction is another core concern of current dispatch studies. Some recent studies discussed reliable estimation of arrival time by using supervised learning on deliveries \cite{zhu2020order,hildebrandt2020supervised}. There also exist works pursuing improving customer's satisfaction by designing reasonable range of customers to whom each restaurant can expose  \cite{ding2020delivery}.

Although the labor rights of riders have attracted sufficient social attention, there are only limited studies on the topics, most of which focus on the impacts of wage design on welfare distribution. For instance, \cite{benjaafar2020labor,wu2020two} examine how wages in platforms impact customers and laborers but have not discussed how the process of dispatching may impact labor welfare. To the best of our knowledge, there is still a lack of systematic discussion about improving the dispatch for mitigating the time-limit stress on riders, which forces riders to drive riskily during delivery orders. We also found that discussion about the impacts of information on labor rights is rare. \cite{mao2019faster,liu2020time} integrate the estimation of travel times, attitudes of customers toward late orders and familiarity of riders with roads into optimization models and show by simulations that data can improve dispatch efficiency. Nevertheless, to the best of our knowledge, there is yet no work discussing the impact of restaurant data and providing theoretical insights into how restaurant data can help improve dispatch efficiency and labor welfare.

\subsection{Labor rights of riders in food-delivery platforms}
Labor rights can be divided into two categories: 1) working benefits such as minimum wages and reasonable working hours; 2) the safety of working conditions\cite{laborright}. Most current relative discussions in the context of platform economy focus on the first economic aspect of labor rights \cite{de2015rise, benjaafar2020labor, parrott2018earnings}, while the working safety of riders is absent from the current literature about platform economy.

Here, our work provides a first step to consider the driving safety of riders in platform economy by minimizing the waiting time of riders. We focus on the waiting time of riders in restaurants for two reasons. First, minimizing the waste of time during the order-picking-up process allows riders to have sufficient time to deliver the order. Second, in contrast to user-centric or profit-driven dispatch, minimizing the rider's waiting time in restaurants can prevent the algorithm from continuously shortening the available time window of deliveries. We believe that our work could offer a new angle for the control community to think about platform economy.

\subsection{Contributions}
We summarize our three contributions below.
\begin{itemize}
\item \textbf{Framework of modeling labor-right protecting problems in meal-delivery platforms.} In contrast to the literature, we developed a model framework addressing labor rights and specifically labor safety in platform economy. As riders play the roles linking all other players on the platform, the labor-right protecting model has to comprehensively model interactions between consumers, restaurants, riders, and the platform.
\item \textbf{Optimal dispatch algorithm for labor rights protection}. Based on the above framework, we propose an efficient dispatch algorithm and rigorously prove that it minimizes the time waste of riders with satisfactory user experience.
\item \textbf{Value of restaurant data on labor rights protection.} Finally, our analysis discovers the sharp value of the private information of restaurants in improving waiting time of riders. Such a result reveals a new perspective on the role of restaurants in protecting labor welfare.
\end{itemize}
\section{Model setting}
In this section, we provide the detailed setting of an online meal delivery platform. Particularly, the environment enrolls four elements: a platform, a restaurant, arrivals of food orders and delivery riders.

\subsection{Model of a Meal-Delivery Process}
In online meal delivery service, an order needs to go through the following process until its corresponding customer can receive the food. First, immediately after being placed online (via apps or websites), the order joins into the processing queue at the restaurant. We assume that the restaurant prepares orders in the processing queue one by one in a first come first serve manner, and it can only prepare one order at one time. After an order is done, the order enters a new waiting queue and waits there until a rider picks it up. In the meantime, the platform will consistently dispatch riders to the restaurant at possibly time-varying rates. In our setting, we assume that one arriving rider will take the earliest prepared order from the waiting queue. If there are no prepared orders, riders will also wait in a queue. When one order is ready, the earliest arriving rider will take this order and leave.

We assume that orders arrive in a Poisson process with rate $1$, and its service time is an exponential distribution of rate $\mu$ independent of the arrival process. The assumption of a Poisson arrival process and exponential service time follows the vast study of restaurant service systems \cite{chen2019food, hwang2010joint, chen2020courier}.
Note that the assumption of a unit-rate arrival process is without loss of generality because we can always normalize the time scale. Denote the food processing queue as $\mathcal{Q}_1$, and the waiting queue of prepared food as $\mathcal{Q}_2$. Then, $\mathcal{Q}_1$ itself is an M/M/1 queue. When one order leaves $\mathcal{Q}_1$, it immediately joins $\mathcal{Q}_2$ to wait for riders. Let $Q_1(t), Q_2(t)$ be the number of orders in $\mathcal{Q}_1, \mathcal{Q}_2$ at time $t$, respectively. When there is no waiting order but $a$ waiting riders, we would write $Q_2(t) = -a.$ In this sense, riders are like negative arrivals that can cancel an order from the queue. Figure \ref{fig:order_procedure} provides an illustration of the whole process.

\begin{figure}[bp]
    \centering
    \includegraphics[scale=0.25]{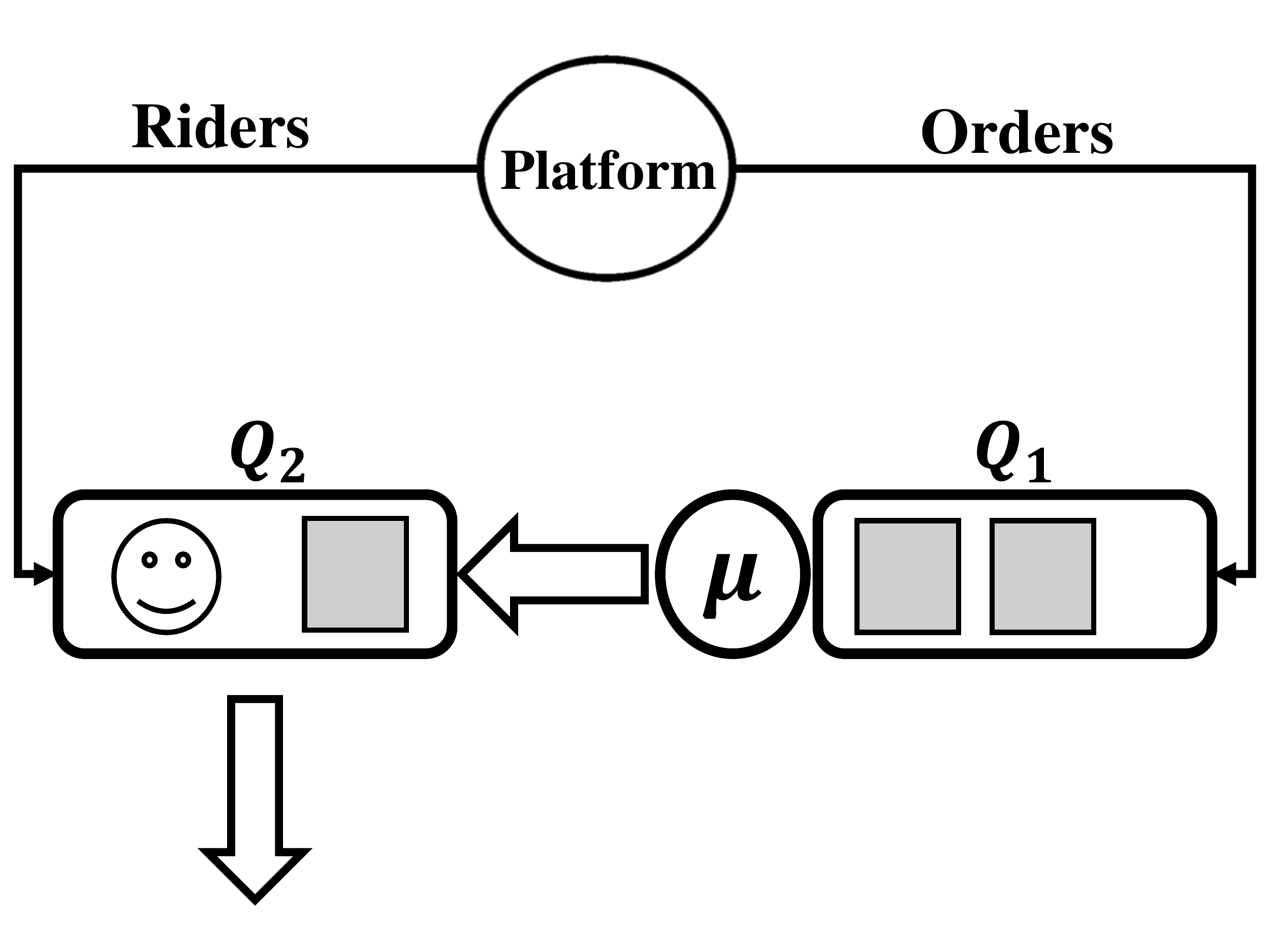}
    \caption{Illustration of Meal Delivery Services}
    \label{fig:order_procedure}
\end{figure}

\subsection{Information Disclosure Dispatch Process}
Based on knowledge of the order arrival rate, the platform must design corresponding dispatch algorithms to send riders. In addition, the platform has two sources of information, one from users and the other from the restaurant, both of which may play a role in the dispatch process. Specifically, we divide the decision process into two parts: 1) users' and restaurant's information disclosure scheme and 2) dispatch process of the platform.
\subsubsection{\textbf{Information Disclosure}}
When a user submits an order to the platform, he/she can reveal his/her patience time $T^*$, that is, how much time he/she can endure for extra order delays due to deliveries. Users' patience time is crucial for riders since in practice, riders would get low tips or even be disqualified for late deliveries \cite{doordash}. We assume that all orders will fully disclose the same fixed value $T^*$. The platform has to design algorithms satisfying this constraint. We formalize this decision problem in Section \ref{sec:labor-problem}.

From the view of the restaurant, when $Q_2(t) \geq 0$, its value is only visible to the restaurant, not the platform, since the platform has no information of whether one order is finished or not. Therefore, the restaurant could decide whether to reveal such information to the platform. To model the restaurant's willingness to share its private data with the platform, we design the following information disclosure mechanism. In particular, the restaurant sets a nonnegative integer $M$ as a threshold that is also known by the platform. For every time $t$, the restaurant maintains a signal $\theta(t)$ visible to the platform, given by
\begin{equation}
\theta(t) = \left\{
\begin{aligned}
0, & ~\text{if }Q_2(t) \leq M;\\
Q_2(t)-M, & ~\text{if }Q_2(t) > M.
\end{aligned}
\right.
\end{equation}
We allow $\theta(t) \equiv 0$, and write it as $M = \infty$. Such a screening policy is motivated by the fact that restaurants would like to hide the number of waiting prepared orders when it is small. In this way, the platform may send excessive riders, which can help reduce the waiting time of orders. The threshold $M$ measures how much information is being shared. When $M = 0$, the restaurant shares all of its private information. However, if $M = \infty$, the restaurant would provide no data to the platform.

\subsubsection{\textbf{Dispatch process}} Finally, we define the dispatch process of the platform. The dispatch policy is assumed to be a time-varying Poisson arrival process of riders whose rates change according to $\theta(t)$ provided by the restaurant. Extending the dispatch policy to more general processes may entail more technical challenges, which exceeds the focus of this paper.

In particular, the platform first sets a list of dispatch rate parameters $\vec{\lambda}=(\lambda_0,\lambda_1,\cdots)$ where $\lambda_i \leq \Lambda$ and the upper bound $\Lambda$ measures the maximal availability of riders near the restaurant. After choosing $\vec{\lambda}$, the platform will send riders to the restaurant in a Poisson process whose rate is given by $\lambda_{\theta(t)}.$ That is, conditioned on the event that the restaurant signal is equal to a value $i$, riders will arrive to the restaurant in a Poisson process of rate $\lambda_{i}$. Furthermore, if there are already $d$ riders waiting at the restaurant, i.e., $Q_2(t) = -d$, the platform will not send any riders even when it is going to do so. We call $d$ the buffer length. 
For ease of exposition, we assume a rider will arrive at the restaurant immediately when he/she is summoned by the platform. In addition, we assume that the arrival process of riders is independent of the order arrival process and service time of orders.

We note here that $\vec{\lambda}, d$ and $\Lambda$ are tightly connected to two crucial parts in today's meal deliveries: postponed dispatch \cite{ulmer2020restaurant} and delivery scope design \cite{ding2020delivery}. Specifically, in postponed dispatch, the platform will not immediately dispatch a rider to restaurants when an order arrives so that they will not arrive with orders being unprepared. Our arrival rates $\vec{\lambda}$ can thus be seen as how fast the platform will send out a rider. Furthermore, when there are already too many riders ($d$ in our model) waiting for food, the platform won't send any rider. In addition, the platform in general will divide a whole city into multiple regions and deploy different numbers of riders in each area. The number of riders in one region will then decide the magnitude of the parameter $\Lambda$.

To ensure that the platform's capacity is sufficient to serve incoming orders and that the system is stable, we impose the following assumptions on dispatch rates of riders.
\begin{assumption}\label{as:stable}
Rates $\mu, \Lambda, \lambda_0$ are larger than $1$, and $\lambda_i > 0$ for all $i > 0$. In addition, it satisfies that there exists a constant $C<\infty$, such that
$
\sum_{i=1}^{\infty}\prod_{j=1}^i \frac{1}{\lambda_i} = C.
$
\end{assumption}

Under the above assumption, the system state can be the pair of $(Q_1(t), Q_2(t))$, which forms a continuous-time Markov chain. Fig. \ref{fig:transition} illustrates a snapshot of this Markov chain. Furthermore, the system is positively recurrent and enjoys a unique stationary distribution based on Assumption \ref{as:stable} \cite{harchol2013performance}. To measure customer experience and labor rights, this paper focuses on the mean order waiting time $\expect{T_o}$ and the mean rider waiting time $\expect{T_r}$ in the system. To be more specific, the waiting time of an order is measured as the time difference between its arrival at the restaurant and its time fetched by a rider. The mean order waiting time is the time-averaged waiting time among all orders. Similarly, the waiting time of one rider is the time between his/her arrival at the restaurant and the time he/she fetches one order. The mean rider waiting time is then the time-averaged waiting time of all riders.

Note that the two expectations are determined by the platform's setting of $\vec{\lambda}, d, M$. Therefore, when we want to stress such dependence, they are written as $\expectsub{\vec{\lambda}, d, M}{T_o}$ and $ \expectsub{\vec{\lambda}, d, M}{T_r}$.
\begin{figure}
    \begin{center}
    \includegraphics[width=3.5in]{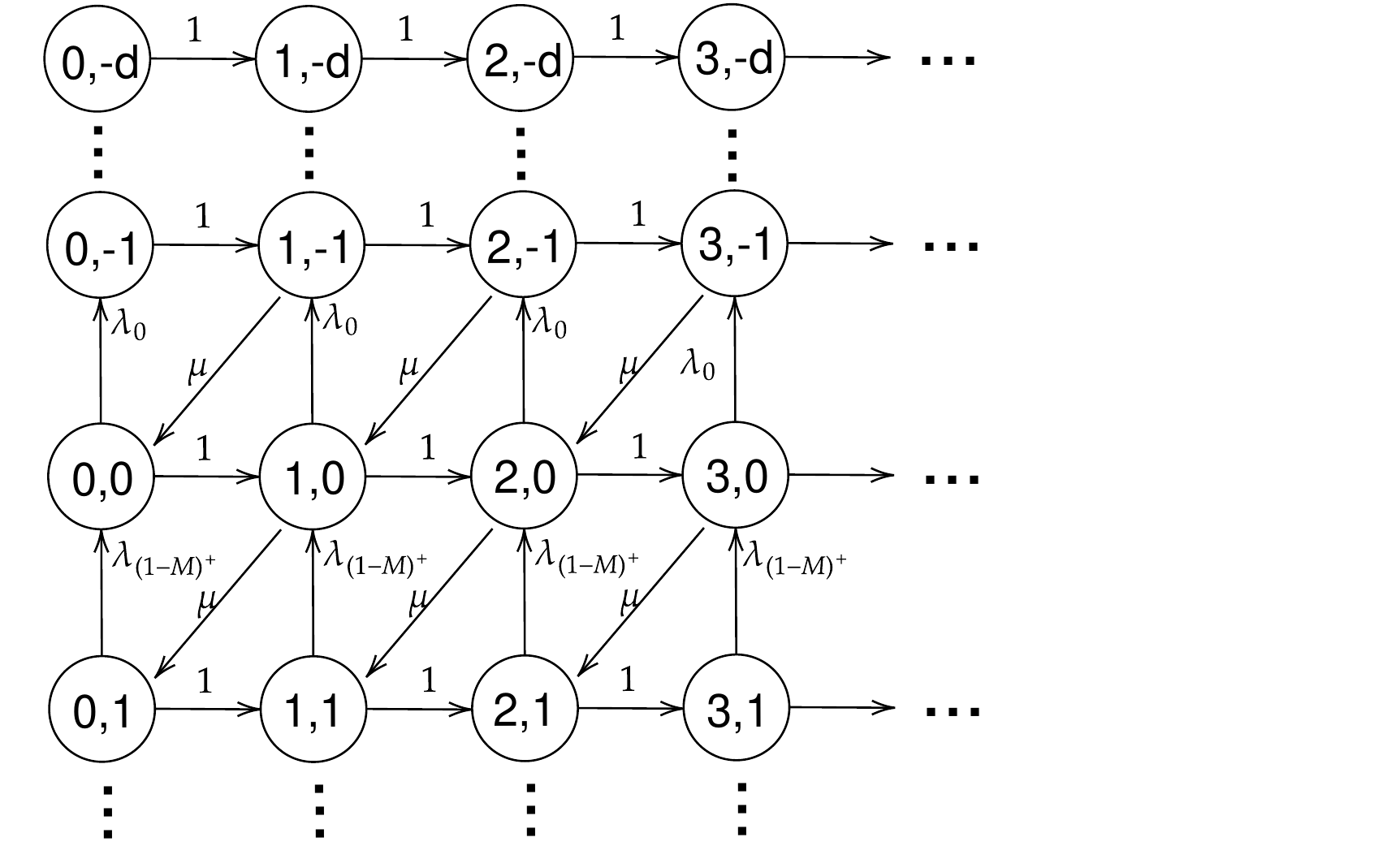}
    \caption{The Markov chain with state $(q_1,q_2)$}
    \label{fig:transition}
    \end{center}
\end{figure}

\section{Labor-Right Protecting Dispatch Problem and Analysis}\label{sec:labor-problem}
Based on the mathematical model of a meal delivery platform, we propose the labor-right protecting dispatch problem in this section. The goal is to protect laborers while guaranteeing customer experience. In particular, this section includes two parts. In the first part, we formulate the labor rights-protecting dispatch problem. In the second part, we analyze the characteristics of the model, which enables us to design algorithms for the dispatch problem in Section \ref{sec:algorithm}.
\subsection{Problem Formulation}
To formalize the experience guarantee, we first notice that the expected order waiting time is at least the mean response time in an $M/M/1$ queue of arrival rate $1$ and service rate $\mu$. Then, we know $\expect{T_o}$ is at least $\frac{1}{\mu - 1}$ \cite{harchol2013performance}.
\begin{proposition}
The averaged order waiting time $\expect{T_o}$ is at least $\frac{1}{\mu - 1}$.
\end{proposition}
Based on this lower bound on expected order waiting time, we can formalize the meaning of patient time $T^*$ in the sense that we require $\expect{T_o} \leq \frac{1}{\mu - 1} + T^*$ since $\expect{T_o} - \frac{1}{\mu - 1}$ measures the extra delays brought by deliveries. Then, after the restaurant reveals its threshold $M$, the platform needs to find the best $\vec{\lambda},d$ to minimize the averaged rider waiting time while satisfying requirements of users.

Based on the above discussion, we propose the following optimization framework for the labor-right protecting dispatch problem.

\begin{equation}\label{eq:optimization}
\begin{aligned}
& \underset{\vec{\lambda},d}{\text{minimize}}
& & \expectsub{\vec{\lambda},d,M}{\text{rider waiting time }T_r} \\
& \text{subject to}
& & \expectsub{\vec{\lambda},d,M}{\text{order waiting time }T_o} \\
& & &\leq \frac{1}{\mu - 1} + \text{user patience time }T^* .
\end{aligned}
\end{equation}
Note that our optimization problem also belongs to the vast domain of constrained Markov decision processes (CMDPs) \cite{altman1999constrained}. The goal of CMDP is to obtain the optimal state-dependent policy of a certain objective while satisfying the given constraints. The optimal policy can be complicated and even randomized. However, in the context of our paper, we not only want to solve the optimization problem (\ref{eq:optimization}) but also seek to explicitly understand the trade-off between the objective "rider waiting time" and the constraint "order waiting time" and how the decisions and the trade-off could change in response to the volume of information provided by the restaurant. Therefore, the policy space in our optimization problems is kept as simple as possible to discuss the above information-related decisions and does not involve the general space of randomized policies. We believe a new framework is needed for such discussions and thus leave it as a future study.

\subsection{Model Analysis}
To obtain a deeper understanding of the problem, we first study the stationary distribution of the system in the general setting where the restaurant allows a threshold $M$, and the platform sets the associated dispatch rate $\vec{\lambda}$ and the buffer length $d$. Based on the distribution, we observe a decoupling phenomenon between $\mathcal{Q}_1$ and $\mathcal{Q}_2$. Finally, we restrict our scope to the case of $M=\infty$ and obtain close-form formulas for $\expect{T_o}$ and $\expect{T_r}$.
\subsubsection{\textbf{Stationary Distribution of the General Model}}
Recall that the state of the system is given by $(Q_1(t), Q_2(t))$, namely, the number of orders waiting for preparation and the number of prepared orders waiting for riders. If $Q_2(t) < 0$, it is the negative of waiting riders. Let $(\mathbf{Q}_1,\mathbf{Q}_2)$ be the corresponding random variables distributed with the stationary distribution of $(Q_1(t), Q_2(t))$. Let $\pi_{q_1,q_2}$ be the stationary distribution of the event $(\mathbf{Q_1} = q_1, \mathbf{Q_2} = q_2)$. Fix the threshold $M$, the dispatch rates $\vec{\lambda} = (\lambda_0,\lambda_1,\cdots)$, and the buffer level $d$, and assume they satisfy the assumption \ref{as:stable}. The state space of the system would be $\mathbb{N} \times \left(\{-d, d+1,\cdots, 1\}\cup \mathbb{N}\right)$.

Define $\rho = \frac{1}{\lambda_0}$. The following lemma shows that$\pi_{q_1,q_2}$ exhibits a product-form property similar to the property in queueing networks with negative arrivals \cite{henderson1993queueing}. The proof of this lemma can be found in the appendix.
\begin{lemma}\label{lemma:distribution}
It holds that for $q_1 \geq 0, q_2 \geq -d$,
\[
\pi_{q_1,q_2} = \left\{
\begin{aligned}
&C_1\left(1-\frac{1}{\mu}\right)\left(\frac{1}{\mu}\right)^{q_1}\rho^{q_2},~\text{if }q_2 \leq M; \\
&C_1\left(1-\frac{1}{\mu}\right)\left(\frac{1}{\mu}\right)^{q_1}\rho^{M}\prod_{i=1}^{q_2-M}\frac{1}{\lambda_i},~\text{otherwise,}
\end{aligned}
\right.
\]
where
\begin{equation}
C_1 = \frac{1-\rho}{\rho^{-d}-\rho^{M+1}+(1-\rho)\rho^M\sum_{i=1}^{\infty}\prod_{j=1}^i \frac{1}{\lambda_i}}.
\end{equation}
\end{lemma}
\subsubsection{Queue Decoupling}\label{sec:queue-decoupling}

Observe that by Lemma \ref{lemma:distribution}, we can indeed decouple $\mathcal{Q}_2$ from $\mathcal{Q}_1$. If we consider only the distribution of events $\mathbf{Q}_2 = q$ for $q \geq -d$, the stationary distribution of $\mathcal{Q}_2$ is equivalent to the stationary distribution of the following Markov chain. In the equivalent Markov chain, the set of the state is $\{-d,\cdots,-1\} \cup \mathbb{N}$. For each state $q \geq -d$, it transitions to state $q + 1$ with rate $1$ and to $q - 1$ with rate $\lambda_{\max(q-M,0)}$ if $q > -d$. The stationary distribution of this Markov chain, denoted by $\nu_q$, is exactly given by
\begin{equation}
\nu_q = \left\{
\begin{aligned}
&C_1\rho^{q_2},~\text{if }q_2 \leq M; \\
&C_1\rho^{M}\prod_{i=1}^{q_2-M}\frac{1}{\lambda_i},~\text{otherwise,}
\end{aligned}
\right.
\end{equation}
where $\rho,C_1$ are the same as those in Lemma \ref{lemma:distribution}.

In addition, let $\mathbf{Q}^*$ be the random variable corresponding to the state in this Markov chain. Then, by equivalence, we have the following property connecting $\expect{T_o},\expect{T_r}$ with $\mathbf{Q}^*$. Due to space limit, we omit its proof.
\begin{lemma}\label{lemma:equivalence}
It holds that
\begin{equation}
\begin{aligned}
\expect{T_o} &= \frac{1}{\mu - 1} + \expect{\max(\mathbf{Q}^*,0)} \\
\expect{T_r} &= -\expect{\min(\mathbf{Q}^*,0)}.
\end{aligned}
\end{equation}
\end{lemma}

\subsubsection{\textbf{Analysis for $M = \infty$}}
Restricting our scope to $M = \infty$, based on Lemma \ref{lemma:distribution}, we can obtain expressions of $\expect{T_o}$ and $\expect{T_r}$ via Little's Law\cite{harchol2013performance}. Formally, we have this lemma.
\begin{lemma}\label{lemma:meanValue}
For fixed $\lambda_0 > 1$ and $d$, let $\rho = \frac{1}{\lambda_0}$. It satisfies that
\begin{subequations}
\begin{equation}
\expectsub{\lambda_0,d}{T_o} = \frac{1}{\mu - 1} + \frac{\rho^{d+1}}{1-\rho},
\end{equation}
\begin{equation}
\expectsub{\lambda_0,d}{T_r} = d - \frac{\rho-\rho^{d+1}}{1-\rho}.
\end{equation}
\end{subequations}
\end{lemma}

\section{Labor-right protecting algorithm} \label{sec:algorithm}
Since it is generally nontrivial to solve (\ref{eq:optimization}), in the first part of this section, we start our analysis by assuming that there is no restaurant information, i.e., $M = \infty$. Under this assumption, we propose Algorithm \ref{algo:optimal-parameter} that provides the procedure to find the optimal $\lambda_0, d$. Second, when there are restaurant data, we design Algorithm \ref{algo:policy-improve} to improve policies based on new information.

\subsection{Algorithm and its Optimality when No Restaurant Data}
When there is no restaurant information, the platform needs only to seek the optimal design of $\lambda_0, d$. The closed-form expressions of $\expect{T_o}$ and $\expect{T_r}$ shown in the Lemma \ref{lemma:meanValue} provides the basis of the following algorithm to solve (\ref{eq:optimization}).

\begin{algorithm}[Optimal Dispatch Algorithm]\label{algo:optimal-parameter}
The algorithm takes $T^*$ and $\Lambda$ as inputs and outputs two parameters $\lambda_0, d$ as follows.
\begin{itemize}
\item If $\frac{\Lambda}{1-\Lambda} \leq T$, simply returns $d = 0, \lambda_0 = \Lambda$;
\item Otherwise, for every $\rho \in [0,1)$, define $D(\rho)$ as the smallest nonnegative integer such that $\frac{\rho^{D(\rho)+1}}{1-\rho} \leq T^*$; and for every nonnegative integer $n$, define $\gamma(n)$ as the maximum real number in $(0,1)$ such that $\frac{\gamma(n)^{n+1}}{1-\gamma(n)} \leq T^*$.
\item Let $d_1 = D(\frac{1}{\Lambda}).$ The algorithm returns $d_1, \frac{1}{\gamma(d_1)}$.
\end{itemize}
\end{algorithm}

The next theorem supports the optimality of Algorithm \ref{algo:optimal-parameter}.
\begin{theorem}\label{thm:optimal-algo}
Under the two parameters $\lambda_0, d$ found by Algorithm \ref{algo:optimal-parameter}, and assuming $M = \infty$, it satisfies $\expect{T_o} = \frac{1}{\mu - 1} + T^*$, and $\expect{T_r}$ is minimized.
\end{theorem}
Let's look at a proof sketch of Theorem \ref{thm:optimal-algo}. A formal proof is provided in the appendix. The goal is to minimize $\expect{T_r}$ while satisfying the delay constraint in (\ref{eq:optimization}). By Lemma \ref{lemma:meanValue}, we want to find $\lambda_0, d$ such that $d-\frac{\rho-\rho^{d+1}}{1-\rho}$ is minimized and $\frac{\rho^{d+1}}{1-\rho} \leq T^*$. The intuition of Algorithm \ref{algo:optimal-parameter} is to choose the smallest $\lambda_0$ satisfying the constraint in (\ref{eq:optimization}). Note that $\expect{T_r}$ is a decreasing function of $\rho$ when $d$ is fixed. Therefore, after finding the smallest $d$, we would increase the value of $\rho$ (or equivalently decrease $\lambda_0$) until the constraint in (\ref{eq:optimization}) is violated. 
\subsection{Policy Improvement using Restaurant Information}
In this section, we explore how to improve platform policies using restaurant data. Particularly, the restaurant provides more information to the platform by reducing its threshold from $M$ to $m$, where $m < M$. Our results highlight that restaurant data are of particular importance to help the platform make better decisions.

Formally, let us assume that the platform has a policy given by when the threshold is $M$. Then, algorithm \ref{algo:policy-improve} outputs a new policy for a given threshold $m, m < M$ that is arguably better than the original policy. Algorithm \ref{algo:policy-improve} is given below.

\begin{algorithm}[Policy Improvement with Restaurant's Data]\label{algo:policy-improve}
The algorithm takes $\vec{\lambda}=(\lambda_0,\cdots),d,M,m$ and $\Lambda$ as inputs and outputs the dispatch rates and buffer level, $\vec{\tau}=(\tau_0,\cdots), d_1$, of the new policy as follows. Define $\rho = 1/\lambda_0$.
\begin{itemize}
\item Set $C = \frac{\rho-\rho^{M-m+1}}{1-\rho}+\rho^{M-m}\sum_{i=M+1}^{\infty}\prod_{j=1}^{i-M} \frac{1}{\lambda_j}.$
\item Set $\tau_0 = \lambda_0, \tau_1 = \frac{1}{(1-\frac{1}{\Lambda})C}$, and $\tau_i = \Lambda$ for all $i > 1$.
\item Let $d_1 = d$. Output $(\vec{\tau}, d_1)$.
\end{itemize}
\end{algorithm}

Our next result shows that using Algorithm \ref{algo:policy-improve}, the platform can always find a policy with better averaged order waiting time and the same averaged rider waiting time based on any policy of a larger threshold.
\begin{theorem}\label{thm:data-value}
Suppose that the platform has a policy with parameters $\vec{\lambda}=(\lambda_0,\cdots),d$ and the restaurant's threshold is $M$, $M > 0$. If $M$ reduces to $m$, $m < M$, then using Algorithm \ref{algo:policy-improve}, the platform can construct a new set of dispatch rates $\vec{\tau}$, such that
\begin{equation}\label{eq:policy-improve}
\expectsub{\vec{\tau},d,m}{T_o} \leq \expectsub{\vec{\lambda},d,M}{T_o}, ~\expectsub{\vec{\tau},d,m}{T_r} = \expectsub{\vec{\lambda},d,M}{T_r},
\end{equation}
where the inequality becomes equal only if
\[
\rho=\frac{\Lambda-1}{\Lambda}\left(\frac{\rho-\rho^{M-m+1}}{1-\rho}+\rho^{M-m}\sum_{i=M+1}^{\infty}\prod_{j=1}^{i-M} \frac{1}{\lambda_j}\right).
\]
\end{theorem}
\begin{remark}
Note that if $M = \infty$ in Theorem \ref{thm:data-value}, the inequality in (\ref{eq:policy-improve}) becomes equal if and only if $\lambda_0 = \Lambda$, that is, the platform is already dispatching riders at its full capacity.
\end{remark}
The proof of Theorem \ref{thm:data-value} can be found in the appendix.
\section{Obligation of different stakeholders}
In this section, we discuss obligations in protecting the labor rights of the three stakeholders: the platform, users, and the restaurant.

First, for restaurant information, Theorem \ref{thm:data-value} tells us that whenever there are data from the restaurant side, the platform can improve the experience of customers while maintaining the same averaged rider waiting time. In this sense, restaurant data help to improve labor rights by making the labor more efficient. We will further quantify the effect of queueing data of the restaurant in Section \ref{sec:simulation}.

Second, we desire to understand how the customer's requirement of experience and the platform's ability to recruit riders may affect the averaged rider waiting time. For ease of exposition, we assume that $M = \infty$ in this part. Furthermore, suppose that the platform is still trying to solve the labor-right protecting dispatch problem introduced in (\ref{eq:optimization}).

In this case, we show that there is a fundamental lower bound of $\expect{T_r}$ given that $T^*$ is low. The proof is deferred to the appendix.
\begin{theorem}[Dilemma between customers and riders]\label{thm:tradeoff}
Suppose that $\hat{\rho}\colon=\frac{1}{\Lambda} \geq 0.3$ and $\frac{\hat{\rho}^4}{1-\hat{\rho}} \geq T^*$ and $M = \infty$. Then, it holds
$
\expect{T_r} \geq \frac{\ln(4/3)}{4}\frac{1}{1-\hat{\rho}}-T^*,
$
for all feasible $d$ and $\lambda_0$ satisfying the constraint in (\ref{eq:optimization}).
\end{theorem}
This result shows that the sum of $T^*$ and $\expect{T_r}$ is lower bounded by a value determined by the platform capacity $\Lambda$. As a result, when $\Lambda$ is fixed, no matter how the platform acts, it is impossible for both $T^*$ and $\expect{T_r}$ to become negligible simultaneously. Nevertheless, the lower bound decreases when $T^*$ or $\Lambda$ increases, thus implying that if customers allow a longer delivery time or if the platform is able to recruit more riders, then laborers can enjoy a lower average waiting time.
\section{Numerical Case} \label{sec:simulation}
In this section, we aim to illustrate the importance of the labor rights protection problem, the dilemma between customer experience and rider waiting time, and the value of restaurant information using simulations. Particularly, we assume that $\lambda_0 = 1$, and $\mu = \Lambda = 1.5$. It is a reasonable assumption because in practice, the restaurant and the platform usually only recruit nearly sufficient numbers of riders to reduce cost.
\subsection{Importance of labor rights protection and Trade-offs}
To study the impact of the labor rights protection problem, we consider three settings of patient time, $T^*=0.01, 0.05, 0.1$. In addition, we assume that there is no restaurant information, so $M = \infty$. Then, for $\lambda_0 \in (1,\Lambda)$, we plot the curve of $\expectsub{\lambda_0,d}{T_r}$, where $d$ is the minimum integer such that $\expectsub{\lambda_0,d}{T_r} - \frac{1}{\mu - 1} \leq T^*$. The result is shown in Fig. \ref{fig:right-protect}.
\begin{figure}[!hbtp]
    \centering
    \includegraphics[scale=0.3]{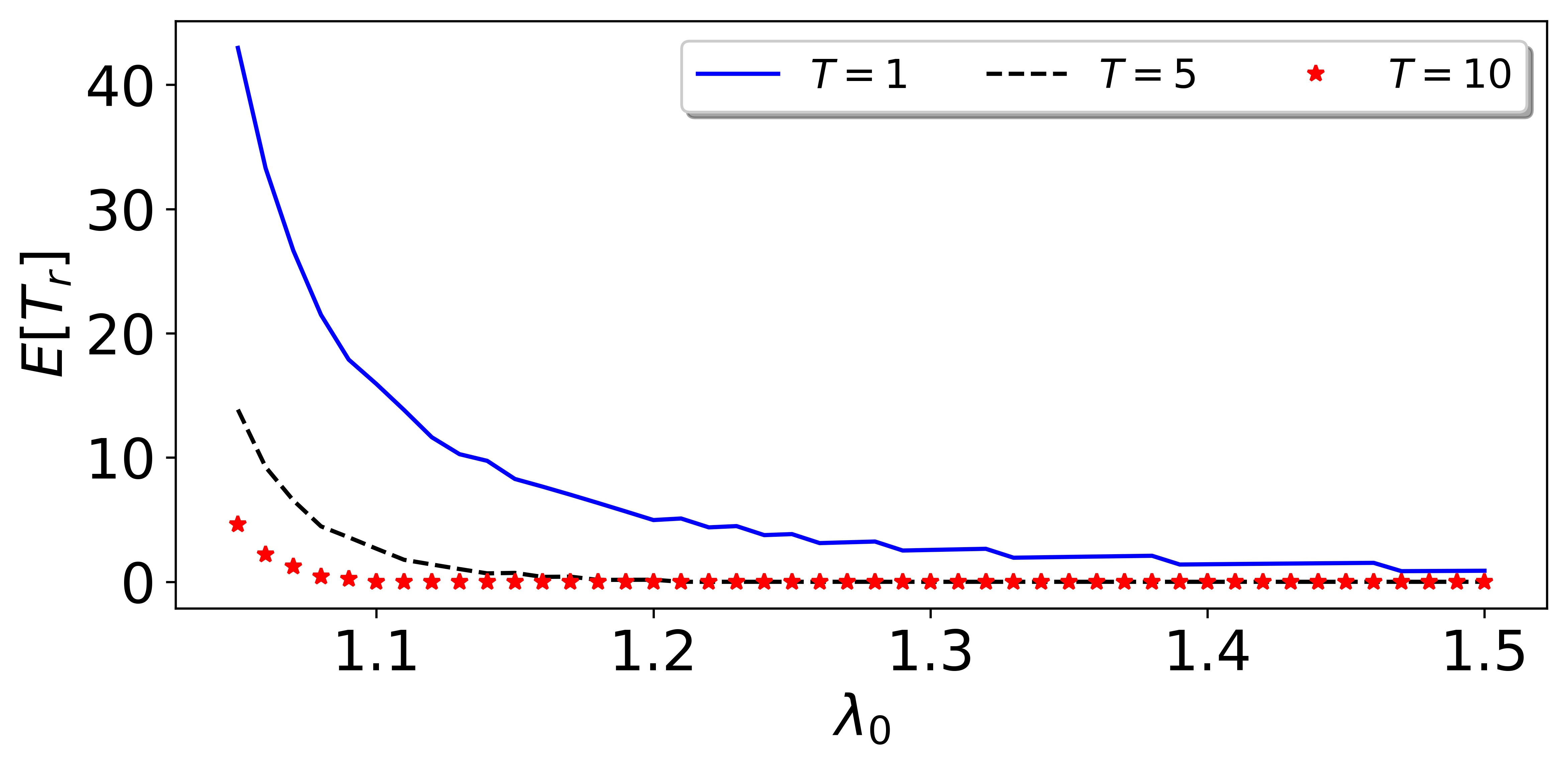}
    \caption{Averaged Rider waiting time under Different Dispatch Rates Given that Customer Experience is satisfied}
    \label{fig:right-protect}
\end{figure}
As we can see, for a fixed $T^*$, although all the patience time is satisfied, the averaged rider waiting time varies greatly. A clever setting of $\lambda_0$ can reduce the waiting time by a factor of $4$. In addition, when $T^*$ increases, the averaged rider waiting time decreases. It characterizes the trade-off between customer experience and labor rights.
\subsection{Value of Restaurant's Data}
To highlight the value of the data from the restaurant, we assume that $M$ could be $0,10,\infty$. In this case, we are able to differentiate different levels of restaurant information. For the platform's policy, we fix $d = 0$ and then allow $\lambda_0$ to vary from $1$ to $\Lambda$ to be the policy for $M = \infty$. Then, for other $M$, we employ Algorithm \ref{algo:policy-improve} to generate new policies. We plot the value of $\expectsub{\vec{\lambda},d,M}{T_o}$ after the policy improvement as shown in Fig. \ref{fig:policy-improve}.
\begin{figure}[!hbtp]
    \centering
    \includegraphics[scale=0.3]{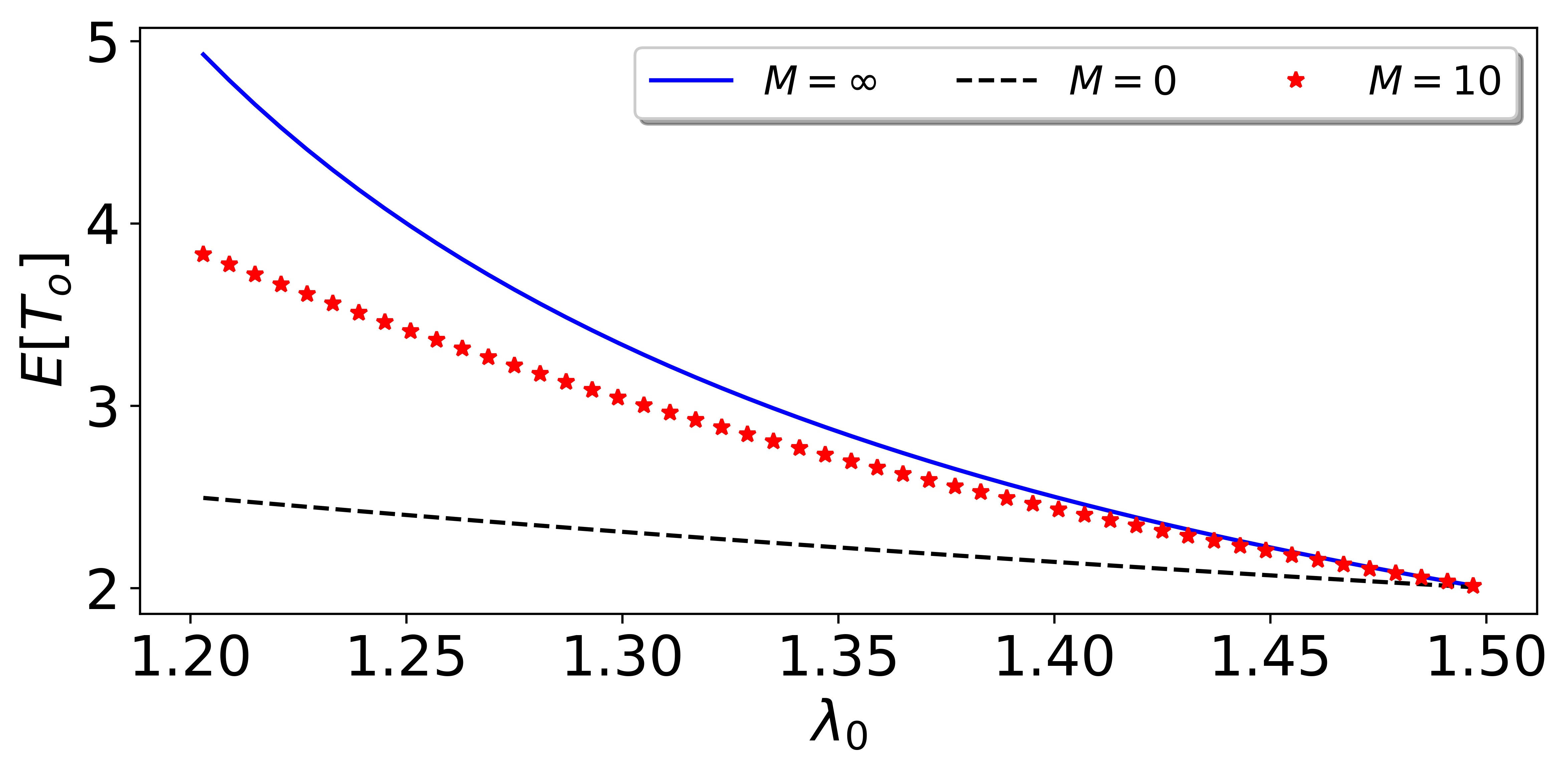}
    \caption{Averaged Order Waiting Time After Using Restaurant Data}
    \label{fig:policy-improve}
\end{figure}
As we can tell from Fig. \ref{fig:policy-improve}, restaurant data and algorithm \ref{algo:policy-improve} can indeed greatly improve the customer experience, decreasing $\expect{T_o}$ by a factor of two when $M = 0$ and $\lambda_0 = 1.2$. Specifically, we can observe that the improvement is greater when $\lambda_0$ or $M$ is smaller.
\section{CONCLUSIONS}
In this paper, we formulated a queueing model for meal delivery platforms that incorporated four key stakeholders: the platform, the restaurant, riders, and customers. In particular, we proposed the labor rights protecting problem, whose goal is to reduce waiting time of riders while maintaining satisfactory customer experience. When there is no restaurant information, we designed an efficient algorithm to find the optimal dispatch parameters for the platform. When the restaurant is willing to share its private data, we further characterized a policy improvement method that takes any policy as an input and returns a refined policy based on restaurant data. The refined policy provably improves customer experience, while labor rights did not get hurt. Our results also indicated that there was a trade-off between customer experience and labor rights in the sense that they cannot be satisfied simultaneously. Our simulations confirmed the theoretical results. Specifically, they revealed that an optimal labor-rights protecting algorithm could dramatically reduce the waiting time of riders, and restaurant data could considerably reduce the order waiting time, especially when the restaurant would like to share more information.


\bibliographystyle{IEEEtran}
\bibliography{references}

\begin{thebibliography}{10}
\providecommand{\url}[1]{#1}
\csname url@samestyle\endcsname
\providecommand{\newblock}{\relax}
\providecommand{\bibinfo}[2]{#2}
\providecommand{\BIBentrySTDinterwordspacing}{\spaceskip=0pt\relax}
\providecommand{\BIBentryALTinterwordstretchfactor}{4}
\providecommand{\BIBentryALTinterwordspacing}{\spaceskip=\fontdimen2\font plus
\BIBentryALTinterwordstretchfactor\fontdimen3\font minus
  \fontdimen4\font\relax}
\providecommand{\BIBforeignlanguage}[2]{{%
\expandafter\ifx\csname l@#1\endcsname\relax
\typeout{** WARNING: IEEEtran.bst: No hyphenation pattern has been}%
\typeout{** loaded for the language `#1'. Using the pattern for}%
\typeout{** the default language instead.}%
\else
\language=\csname l@#1\endcsname
\fi
#2}}
\providecommand{\BIBdecl}{\relax}
\BIBdecl

\bibitem{ChinaNews}
``Takeaway riders become a high-risk profession? two takeaway platforms: set
  flexible time,''
  \url{https://www.tellerreport.com/life/2020-09-10-takeaway-riders-become-a-high-risk-profession--two-takeaway-platforms--set-flexible-time.BJexgaC4DEP.html},
  accessed: 2021-09-06.

\bibitem{byun2017characteristics}
J.~H. Byun, B.~Y. Jeong, and M.~H. Park, ``Characteristics of motorcycle
  crashes of food delivery workers,'' \emph{Journal of the Ergonomics Society
  of Korea}, vol.~36, no.~2, pp. 157--168, 2017.

\bibitem{DangerousJob}
``Top 25 most dangerous jobs in the united states,''
  \url{https://advisorsmith.com/data/most-dangerous-jobs/}, accessed:
  2021-09-06.

\bibitem{DriverRisk}
``Food delivery: Drivers take the risks. platforms reap the rewards.''
  \url{https://technode.com/2019/12/04/food-delivery-drivers-take-the-risks-platforms-reap-the-rewards/},
  accessed: 2021-09-06.

\bibitem{christie2019health}
N.~Christie and H.~Ward, ``The health and safety risks for people who drive for
  work in the gig economy,'' \emph{Journal of Transport \& Health}, vol.~13,
  pp. 115--127, 2019.

\bibitem{bai2019coordinating}
J.~Bai, K.~C. So, C.~S. Tang, X.~Chen, and H.~Wang, ``Coordinating supply and
  demand on an on-demand service platform with impatient customers,''
  \emph{Manufacturing \& Service Operations Management}, vol.~21, no.~3, pp.
  556--570, 2019.

\bibitem{taylor2018demand}
T.~A. Taylor, ``On-demand service platforms,'' \emph{Manufacturing \& Service
  Operations Management}, vol.~20, no.~4, pp. 704--720, 2018.

\bibitem{feldman2019can}
P.~Feldman, A.~E. Frazelle, and R.~Swinney, ``Can delivery platforms benefit
  restaurants?'' \emph{Available at SSRN 3258739}, 2019.

\bibitem{chen2019food}
M.~Chen, M.~Hu, and J.~Wang, ``Food delivery service and restaurant: Friend or
  foe?'' \emph{Available at SSRN 3469971}, 2019.

\bibitem{chen2020courier}
M.~Chen and M.~Hu, ``Courier dispatch in on-demand delivery,'' \emph{Available
  at SSRN}, 2020.

\bibitem{gelenbe1991queues}
E.~Gelenbe, P.~Glynn, and K.~Sigman, ``Queues with negative arrivals,''
  \emph{Journal of applied probability}, vol.~28, no.~1, pp. 245--250, 1991.

\bibitem{harrison1993sojourn}
P.~G. Harrison and E.~Pitel, ``Sojourn times in single-server queues by
  negative customers,'' \emph{Journal of Applied Probability}, vol.~30, no.~4,
  pp. 943--963, 1993.

\bibitem{harrison1996m}
------, ``The m/g/1 queue with negative customers,'' \emph{Advances in Applied
  Probability}, vol.~28, no.~2, pp. 540--566, 1996.

\bibitem{li2004map}
Q.-L. Li and Y.~Q. Zhao, ``A map/g/1 queue with negative customers,''
  \emph{Queueing Systems}, vol.~47, no.~1, pp. 5--43, 2004.

\bibitem{henderson1993queueing}
W.~Henderson, ``Queueing networks with negative customers and negative queue
  lengths,'' \emph{Journal of Applied Probability}, pp. 931--942, 1993.

\bibitem{gelenbe1991product}
E.~Gelenbe, ``Product-form queueing networks with negative and positive
  customers,'' \emph{Journal of applied probability}, vol.~28, no.~3, pp.
  656--663, 1991.

\bibitem{yildiz2019provably}
B.~Yildiz and M.~Savelsbergh, ``Provably high-quality solutions for the meal
  delivery routing problem,'' \emph{Transportation Science}, vol.~53, no.~5,
  pp. 1372--1388, 2019.

\bibitem{reyes2018meal}
D.~Reyes, A.~Erera, M.~Savelsbergh, S.~Sahasrabudhe, and R.~O'Neil, ``The meal
  delivery routing problem,'' \emph{Optimization Online}, 2018.

\bibitem{joshi2020batching}
M.~Joshi, A.~Singh, S.~Ranu, A.~Bagchi, P.~Karia, and P.~Kala, ``Batching and
  matching for food delivery in dynamic road networks,'' \emph{arXiv preprint
  arXiv:2008.12905}, 2020.

\bibitem{steever2019dynamic}
Z.~Steever, M.~Karwan, and C.~Murray, ``Dynamic courier routing for a food
  delivery service,'' \emph{Computers \& Operations Research}, vol. 107, pp.
  173--188, 2019.

\bibitem{ulmer2020restaurant}
M.~W. Ulmer, B.~W. Thomas, A.~M. Campbell, and N.~Woyak, ``The restaurant meal
  delivery problem: Dynamic pickup and delivery with deadlines and random ready
  times,'' \emph{Transportation Science}, 2020.

\bibitem{liu2020time}
S.~Liu, L.~He, and Z.-J. Max~Shen, ``On-time last-mile delivery: Order
  assignment with travel-time predictors,'' \emph{Management Science}, 2020.

\bibitem{liu2020integrating}
S.~Liu, H.~Jiang, S.~Chen, J.~Ye, R.~He, and Z.~Sun, ``Integrating dijkstra’s
  algorithm into deep inverse reinforcement learning for food delivery route
  planning,'' \emph{Transportation Research Part E: Logistics and
  Transportation Review}, vol. 142, p. 102070, 2020.

\bibitem{zhu2020order}
L.~Zhu, W.~Yu, K.~Zhou, X.~Wang, W.~Feng, P.~Wang, N.~Chen, and P.~Lee, ``Order
  fulfillment cycle time estimation for on-demand food delivery,'' in
  \emph{Proceedings of the 26th ACM SIGKDD International Conference on
  Knowledge Discovery \& Data Mining}, 2020, pp. 2571--2580.

\bibitem{hildebrandt2020supervised}
F.~D. Hildebrandt and M.~W. Ulmer, ``Supervised learning for arrival time
  estimations in restaurant meal delivery,'' Working paper, Tech. Rep., 2020.

\bibitem{ding2020delivery}
X.~Ding, R.~Zhang, Z.~Mao, K.~Xing, F.~Du, X.~Liu, G.~Wei, F.~Yin, R.~He, and
  Z.~Sun, ``Delivery scope: A new way of restaurant retrieval for on-demand
  food delivery service,'' in \emph{Proceedings of the 26th ACM SIGKDD
  International Conference on Knowledge Discovery \& Data Mining}, 2020, pp.
  3026--3034.

\bibitem{benjaafar2020labor}
S.~Benjaafar, J.-Y. Ding, G.~Kong, and T.~Taylor, ``Labor welfare in on-demand
  service platforms,'' \emph{Available at SSRN 3103447}, 2020.

\bibitem{wu2020two}
S.~Wu, S.~Xiao, and S.~Benjaafar, ``Two-sided competition between on-demand
  service platforms,'' \emph{Available at SSRN 3525971}, 2020.

\bibitem{mao2019faster}
W.~Mao, L.~Ming, Y.~Rong, C.~S. Tang, and H.~Zheng, ``Faster deliveries and
  smarter order assignments for an on-demand meal delivery platform,''
  \emph{Available at SSRN 3469015}, 2019.

\bibitem{laborright}
``Labor rights,'' \url{https://en.wikipedia.org/wiki/Labor_rights}, accessed:
  2021-09-06.

\bibitem{de2015rise}
V.~De~Stefano, ``The rise of the just-in-time workforce: On-demand work,
  crowdwork, and labor protection in the gig-economy,'' \emph{Comp. Lab. L. \&
  Pol'y J.}, vol.~37, p. 471, 2015.

\bibitem{parrott2018earnings}
J.~A. Parrott and M.~Reich, ``An earnings standard for new york city’s
  app-based drivers,'' \emph{New York: The New School: Center for New York City
  Affairs}, 2018.

\bibitem{hwang2010joint}
J.~Hwang, L.~Gao, and W.~Jang, ``Joint demand and capacity management in a
  restaurant system,'' \emph{European journal of operational research}, vol.
  207, no.~1, pp. 465--472, 2010.

\bibitem{doordash}
``Lateness based violations explained,''
  \url{https://help.doordash.com/dashers/s/article/Lateness-Based-Deactivations-Explained?language=en_US},
  accessed: 2021-09-06.

\bibitem{harchol2013performance}
M.~Harchol-Balter, \emph{Performance modeling and design of computer systems:
  queueing theory in action}.\hskip 1em plus 0.5em minus 0.4em\relax Cambridge
  University Press, 2013.

\bibitem{altman1999constrained}
E.~Altman, \emph{Constrained Markov decision processes}.\hskip 1em plus 0.5em
  minus 0.4em\relax CRC Press, 1999, vol.~7.

\end{thebibliography}

\appendix
\subsection{Proof of Lemma \ref{lemma:distribution}}
\begin{proof}
Notice that $\sum_{q_1=0}^{\infty}\sum_{q_2=-d}^{\infty}\pi_{q_1,q_2} = 1$. We can then prove the result by checking that $\pi_{q_1,q_2}$ satisfies the balance equation of the Markov chain because the steady state distribution is the unique solution to the balance equation given that it is normalized\cite{harchol2013performance}. For any $x$, denote $x^+=\max(0,x)$. In the considered Markov chain, for a state $(q_1,q_2)$, its transitions are given as follows:
\[
(q_1,q_2) \rightarrow \left \{
\begin{aligned}
(q_1+1,q_2) &,~\text{of rate }1; \\
(q_1-1,q_2+1) &,~\text{of rate }\mu~\text{if }q_1 > 0; \\
(q_1,q_2-1) &,~\text{of rate }\lambda_{(q_2-M)^+}~\text{if }q_2 > -d.
\end{aligned}
\right.
\]
For any two state $Q,P$, denote the rate from $Q$ to $P$ by $r_{Q,P}$. Fix a state $Q=(q_1,q_2)$. Here we only check the balance equation for state $(q_1,q_2)$ satisfying $q_1 > 0,q_2 > -d$. The same strategy can verify other cases similarly. The balance equation of this state is that 
\[
\text{inflow}\colon = \sum_{P} \pi_{P}r_{P,Q} = \sum_{P} \pi_Q r_{Q,P} =\colon \text{outflow}.
\]
We need to consider three cases. First, suppose that $q_2 < M$. Then by the form of $\pi$, it holds
\[
\begin{aligned}
\text{inflow}&=\pi_{q_1-1,q_2}+\pi_{q_1,q_2+1}\lambda_0+\pi_{q_1+1,q_2-1}\mu \\
&= C\left(\frac{1}{\mu^{q_1-1}\lambda_0^{q_2}}+\frac{1}{\mu^{q_1}\lambda_0^{q_2}}+\frac{1}{\mu^{q_1}\lambda_0^{q_2-1}}\right) \\
&= \pi_{q_1,q_2}\left(1+\lambda_0+\mu\right)=\text{outflow}.
\end{aligned}
\]
Then if $q_2 = M$, it satisfies that 
\[
\begin{aligned}
\text{inflow}&=\pi_{q_1-1,q_2}+\pi_{q_1,q_2+1}\lambda_1+\pi_{q_1+1,q_2-1}\mu \\
&= C\left(\frac{1}{\mu^{q_1-1}\lambda_0^{q_2}}+\frac{1}{\mu^{q_1}\lambda_0^{q_2}}+\frac{1}{\mu^{q_1}\lambda_0^{q_2-1}}\right) \\
&= \pi_{q_1,q_2}\left(1+\lambda_0+\mu\right)=\text{outflow}.
\end{aligned}
\]
Finally, for $q_2 > M$, 
\[
\begin{aligned}
&\mspace{23mu}\text{inflow} \\
&=\pi_{q_1-1,q_2}+\pi_{q_1,q_2+1}\lambda_{q_2+1-M}+\pi_{q_1+1,q_2-1}\mu \\
&= C\left(\frac{1}{\mu^{q_1-1}\lambda_0^{q_2}}+\frac{1}{\mu^{q_1}\lambda_0^{M}}\left(\prod_{i=1}^{q_2-M}\frac{1}{\lambda_i}+\prod_{i=1}^{q_2-M-1}\frac{1}{\lambda_i}\right)\right) \\
&= \pi_{q_1,q_2}\left(1+\lambda_{q_2-M}+\mu\right)=\text{outflow},
\end{aligned}
\]
which completes the proof.
\end{proof}
\subsection{Proof of Lemma \ref{lemma:meanValue}}
\begin{proof}
Define $\expect{N_o}$ be the expected number of unsent orders. By Little's Law, $\expect{T_o} = \expect{N_o}$ because orders' arrival rate is $1$. It holds that 
\[
\expect{N_o} = \sum_{q_1=0}^{\infty}\sum_{q_2=0}^{\infty} (q_1+q_2)\pi_{q_1,q_2}.
\]
Simplifying Lemma \ref{lemma:distribution} by setting $M = \infty$. We have $\pi_{q_1,q_2} = (1-\rho)\rho^d(1-\frac{1}{\mu})\frac{1}{\mu^q_1}\rho^{q_2}$ for $q_1 \geq 0,q_2 \geq -d$.
Then 
\[
\begin{aligned}
\expect{N_o} &= \sum_{q_1,q_2=0}^{\infty} q_1 \pi_{q_1,q_2} + \sum_{q_1,q_2=0}^{\infty} q_2 \pi_{q_1,q_2} \\
&= \frac{1}{\mu - 1} + \sum_{q_2=0}^{\infty} (1-\rho)\rho^{d+j}j\\
&= \frac{1}{\mu - 1} + \frac{\rho^{d+1}}{1-\rho}.
\end{aligned}
\]
On the other hand, define $\expect{N_r}$ be the expected waiting riders. Again, by Little's Law, $\expect{N_r} = \hat{\lambda}\expect{T_r}$, where $\hat{\lambda}$ is the "true" dispatch rate of riders. That is, 
\[
\hat{\lambda} = \lambda_0\left(1 - \sum_{q_1=0}^{\infty}\pi_{q_1,-d}\right)=\lambda_0\rho = 1.
\]
Then for $\expect{N_r}$, it holds
\[
\begin{aligned}
\expect{N_r} = \sum_{q_1=0}^{\infty}\sum_{q_2=-d}^{-1} -q_2 \pi_{q_1,q_2} &= (1-\rho)\rho^d\sum_{j=1}^d j\rho^j \\
&= d - \frac{1 - \rho^d}{\lambda_0 - 1},
\end{aligned}
\]
and thus $\expect{T_r} = d - \frac{1 - \rho^d}{\lambda_0 - 1}$.
\end{proof}
\subsection{Formal Proof of Theorem \ref{thm:optimal-algo}}

We first have the following lemma demonstrating the optimality of choosing smallest feasible $d$. 

Recall the function $D(\rho)$ and the function $\gamma(n)$ from Algorithm \ref{algo:optimal-parameter} defined for a fixed $T^*$. Let $d_l =max(1,D(\frac{1}{\Lambda}))$. For every $d \geq d_l$, define $h(d) = \expectsub{\frac{1}{\gamma(d)},d}{T_r}$. 
\begin{lemma}\label{lemma:monotonicity}
It holds that $h(d)$ is a strictly increasing function when $d \geq d_l$.
\end{lemma}
\begin{proof}
Let $\rho_l$ be the smallest real number such that $\frac{\rho_l^{2}}{1-\rho_l} \geq T^*$. And if $\rho_l < \frac{1}{\Lambda}$, set it to be $\frac{1}{\Lambda}$. We can then see that $d_l = D(\rho_l)$. Now for every $\rho \geq \rho_l$, define $x(\rho)$ such that $\frac{\rho^{x(\rho)+1}}{1-\rho} = T^*$. Then $x(\rho)+1 = \frac{\ln T^* + \ln(1-\rho)}{\ln \rho} \geq 1$. Note that $x(\rho)$ may not be integer. 

We claim that for every $d \geq d_l$, it holds $x(\gamma(d)) = d$. That is, $x$ is like a reverse function of $\gamma$, but it extends the definitional domain. We prove this claim by contradiction. Suppose it is not. It must hold that $x(\gamma(d)) < d$, and $\frac{(\gamma(d))^{d+1}}{1-\gamma(d)} < T^*$. However, in this case, we can find a $\rho > \gamma(d)$ such that $\frac{\rho^{d+1}}{1-\rho)} \leq T^*$, which contradicts the fact that $\gamma(d)$ is the largest such possible $\rho$. As a result, we must have $x(\gamma(d))=d$.

For $\rho \geq \rho_l$, define $g(\rho)\colon= x(\rho)-\frac{\rho - \rho^{x(\rho)+1}}{1-\rho}$. Then to prove that $h(d)$ is strictly increasing, it is sufficient to prove that $g(\rho)$ is strictly increasing since $h(d) = g(\gamma(d))$. Furthermore, observe that for $\rho \geq \rho_l$, $\frac{\rho^{x(\rho)+1}}{1-\rho} = T^*$, which is a constant. It is enough to prove the function $G(\rho)\colon=\frac{\ln T^* + \ln(1-\rho)}{\ln \rho}-\frac{\rho}{1-\rho}-1$ is strictly increasing.

To prove it, take derivative of $G$, and we can get 
\begin{equation}
\begin{aligned}
G'(\rho) &= \frac{-1}{(1-\rho)\ln \rho}-\frac{\ln T^*+\ln(1-\rho)}{\rho(\ln \rho)^2}-\frac{1}{(1-\rho)^2} \\
&=\frac{-1}{(1-\rho)\ln \rho}-\frac{(x(\rho)+1)\ln \rho}{\rho(\ln \rho)^2}-\frac{1}{(1-\rho)^2} \\
&\geq \frac{-1}{(1-\rho)\ln \rho}-\frac{1}{\rho\ln \rho}-\frac{1}{(1-\rho)^2},
\end{aligned}
\end{equation}
where the last inequality is because $x(\rho)+1\geq 1$ and $\ln(\rho) < 0.$

It remains to show that $G_1(\rho)=\frac{-1}{(1-\rho)\ln \rho}-\frac{1}{\rho\ln \rho}-\frac{1}{(1-\rho)^2}$ is positive in $(0,1)$. Simplifying the terms gives $G_1(\rho)=\frac{-\rho+\rho\ln \rho + 1}{-(1-\rho)^2\rho \ln \rho}$. For the function $G_2(\rho)=-\rho + \rho \ln \rho + 1$, its derivative is $\ln \rho$, which is negative in $(0,1)$. Therefore, the function $G_2(\rho)$ is a strictly decreasing function in $(0,1)$. Since $G_2(1) = 0$, we know both $G_2(\rho)$ is always positive in $(0,1)$, and so is $G_1(\rho)$. Consequently, $G(\rho)$ is a strictly increasing function, which completes the proof.
\end{proof}
We can now finish the proof of Theorem \ref{thm:optimal-algo}.

\begin{proof}
In Algorithm \ref{algo:optimal-parameter}, we first check whether $\frac{\Lambda}{1-\Lambda} \leq T^*$. If it is, it means it is feasible to set $d = 0$, and $\lambda_0 = \Lambda$. It is optimal because $\expectsub{\Lambda,0}{T_r} = 0$.

Otherwise, we have to set $d \geq 1$. Notice that when $d$ is fixed, $\expectsub{\Lambda,0}{T_o}$ is a decreasing function, while $\expectsub{\Lambda,0}{T_r}$ is an increasing function. Therefore, suppose we want to select $d$ one by one. We can first check whether the pair of $(\lambda_0 = \Lambda,d)$ can satisfy the constraint in (\ref{eq:optimization}). If it does, then we decrease the value of $\lambda_0$ to obtain lower averaged rider waiting time. The above discussion claims that feasible $d$ must be at least $D(\frac{1}{\Lambda})$. And for each feasible $d$, the optimal choice is to choose $\lambda_0 = \frac{1}{\gamma(d)}$. The averaged rider waiting time given by this parameter setting is exactly $h(d)$. Now according to Lemma \ref{lemma:monotonicity}, $h(d)$ is an increasing function. Therefore, the optimal solution is to set $d = D(\frac{1}{\Lambda})$, and $\lambda_0 = \frac{1}{\gamma(d)}$.
\end{proof}
\subsection{Proof of Theorem \ref{thm:data-value}}
We now provide the proof of Theorem \ref{thm:data-value}. The intuition is that, 
by Little's Law, it suffices to prove that after using the new policy $(\vec{\tau},d,m)$, the mean number of waiting riders remains the same, but the mean number of waiting orders decreases.

Then the proof starts by coupling the Markov chain of the threshold $m$ and that of the threshold $M$, $M > m$. After that, we propose an optimization problem to minimize the mean number of waiting orders, and use a greedy policy to solve it. Algorithm \ref{algo:policy-improve} corresponds to the final outcome of this policy, and thus makes an improvement.

\subsubsection{System Coupling} 
To enable analysis, recall the discussion in Section \ref{sec:queue-decoupling}, we only need to consider the equivalent Markov chain of $\mathcal{Q}_2$ to calculate the mean numbers of waiting orders and riders. Denote the new Markov chain of $(\vec{\lambda},d,M)$ as $A$, and that of $(\vec{\tau},d,m)$ as $B$. Denote their steady-state distributions by $\pi_A(q), \pi_B(q)$. We can now couple the two Markov chains on state $m$. Recall that $\rho = \frac{1}{\lambda_0}$. By solving balance equations with state $m$ as the initial start, we could notice that
\[
\pi_A(m)\left(\sum_{i=-d}^{M} \rho^{i-m}+\rho^{M-m}\sum_{i=M+1}^{\infty}\prod_{j=1}^{i-M} \frac{1}{\lambda_j}\right) = 1,
\]
and
\[
\pi_B(m)\left(\sum_{i=-d}^{m} \rho^{i-m} + \tau_1\sum_{i=m+1}^{\infty} \left(\frac{1}{\Lambda}\right)^{i-m-1}\right)=1.
\]
Then by the definition of $\tau_1$, it holds that $\pi_A(m) = \pi_B(m)$. As a result, by Lemma \ref{lemma:equivalence}, it leads to the fact that 
\[
\expectsub{\vec{\tau},d,m}{T_r} =\sum_{i=-d}^{-1}(-i)\rho^{i-m}\pi_A(m)= \expectsub{\vec{\lambda},d,M}{T_r}.
\]
\subsubsection{Optimization Framework and Greedy Improvement}
It remains to show that the mean number of waiting orders decreases. For Markov chain $A$, define $a_q = \frac{\pi_A(q)}{\pi_A(m)}$ for $q > 0$. Similarly, define $b_q$ for Markov chain $B$. Indeed, $a_q = \frac{a_{q-1}}{\lambda_{\max(q-M,0)}}$. Then by Lemma \ref{lemma:equivalence}, we have
\[
\expectsub{\vec{\lambda},d,M}{T_o} - \expectsub{\vec{\tau},d,m}{T_o} = \pi_A(m)\sum_{q=1}^{\infty} q\left(a_q - b_q\right).
\]
Furthermore, we can observe that by the definition of $\tau_1$, it holds $\sum_{q=1}^{\infty} a_q = \sum_{q=1}^{\infty} b_q < \infty$. Denote this constant by $H$. To show that $\expectsub{\vec{\lambda},d,M}{T_o} \geq \expectsub{\vec{\tau},d,m}{T_o}$, it suffices to prove $\{b_q\}$ is the solution to the following optimization problem:
\begin{equation}\label{eq:optimize-delay}
\begin{aligned}
& \underset{\vec{x}=x_1,\cdots}{\text{minimize}}
& & \sum_{i=1}^{\infty} i\cdot x_i \\
& \text{subject to}
& & \sum_{i=1}^{\infty} x_i = H \\
& & & \frac{x_i}{x_{i - 1}} \geq \frac{1}{\Lambda}, \forall i > 1 \\
& & & x_1 \geq \frac{1}{\Lambda}.
\end{aligned}
\end{equation}
The requirement of $\frac{x_i}{x_{i - 1}} \geq \frac{1}{\Lambda}$ comes from the fact that the dispatch rate cannot exceed $\Lambda$.

Greedily, given a feasible setting of $\vec{x}$ in \ref{eq:optimize-delay}, the best strategy is to keep $x_i$ large for small $i$, and keep $x_i$ small for large $i$. Particularly, we claim that the solution must satisfy $\frac{x_{i+1}}{x_i} = \frac{1}{\Lambda}$ for all $i \geq 1$. Otherwise, we can find the smallest $i$ with $\frac{x_{i+1}}{x_i} > \frac{1}{\Lambda}$. It always improves the objective by decreasing $x_{i+1}$ and increasing $x_i$. Then we can solve out $x_1$, and it turns out that the produced solution is indeed $\{b_q\}$, which completes the proof.
\hspace*{\fill}~\QED\par\endtrivlist\unskip

\subsection{Proof of Theorem \ref{thm:tradeoff}}
Fix $\lambda_0$ and $d$ satisfying the constraint in (\ref{eq:optimization}). By the assumption in Theorem \ref{thm:tradeoff}, it holds that $\frac{1}{\lambda_0} \geq 0.3$ and $d \geq 3$. Then to show that there is a trade-off between customer experience and labor right, one key observation is that, by Lemma \ref{lemma:meanValue}, 
\begin{equation}
\expect{T_r} + \expect{T_o}-\frac{1}{\mu-1} = d + \frac{2\rho^{d+1}-\rho}{1-\rho},   
\end{equation}
where $\rho = \frac{1}{\lambda_0}$. Let $f(\rho) =  \frac{2\rho^{d+1}-\rho}{1-\rho}$. We then consider two cases.

First, when $\rho^{d+1} \geq \frac{3}{4}$, it holds $f(\rho) \geq \frac{1}{2(1-\rho)}$ because $\rho < 1$. Therefore, $\expect{T_r} + \expect{T_o}-\frac{1}{\mu-1} \geq \frac{1}{2(1-\rho)}.$

Second, when $\rho^{d+1} < \frac{3}{4}$, we claim that $\expect{T_r} + \expect{T_o}-\frac{1}{\mu-1} \geq \frac{d}{2}$. To see why it is true, we only need to show that $f(\rho)=\rho \frac{2\rho^d-1}{1-\rho} \geq -\frac{\rho d}{2}$. For $x \in [0,1)$, Define $u(x) = \frac{2x-1}{1-x^{1/d}}.$ We claim that $u(x) \geq -\frac{d}{2}$. 

\begin{proof}
If not, then it holds that 
\[
\frac{2x-1}{1-x^{1/d}} \leq -\frac{d}{2} \Longleftrightarrow 0 \leq -\frac{d}{2}+\frac{d}{2}x^{1/d}-2x+1.
\]
Define $v(x) = -\frac{d}{2}+\frac{d}{2}x^{1/d}-2x+1$. Take derivative of $v(x)$, we know $v'(x) = \frac{1}{2}x^{1/d-1}-2$. Since $d \geq 3$, $v'(x)$ is a decreasing function. In addition, $\lim_{x \to 0^+} v'(x) > 0$, and $v'(1) < 0$. Therefore, the value $\bar{x}$ such that $v'(\bar{x}) = 0$ is the maximum point of $v(x)$. Set $v'(\bar{x}) = 0$. It holds $\bar{x} = 4^{\frac{-d}{d-1}}$. Therefore, the maximum value of $h(x)$ for $x \in (0,1)$ would be $h(\bar{x}) = -\frac{d}{2}+\frac{d}{2}4^{\frac{1}{d-1}}-2\bar{x}+1 \leq -2\cdot 4^{-1.5}+1<0$. It thus contradicts the assumption that there is some $x$ satisfying $h(x) \geq 0$, which completes the proof.
\end{proof}

As a result, it holds $\expect{T_r}+\expect{T_o}-\frac{1}{\mu - 1} \geq \frac{d}{2} \geq \frac{\ln(4/3)}{-2\ln \rho}$ since $\rho^d \leq \frac{3}{4}$. Furthermore, by assumption, $\rho \geq 0.3$, meaning that $-\ln \rho \leq 2(1-\rho)$. Therefore, summarizing the above discussions, we have
\begin{equation}
\expect{T_r} \geq \frac{\ln(4/3)}{4(1-\rho)}-\left(\expect{T_o}-\frac{1}{\mu-1})\right) \geq \frac{\ln(4/3)}{4(1-\hat{\rho})}-T^*,  
\end{equation}
where $\hat{\rho} = \frac{1}{\Lambda}$.\hspace*{\fill}~\QED\par\endtrivlist\unskip

\end{document}